\documentclass[a4paper,UKenglish,cleveref, autoref]{lipics-v2019}

\title{A complete axiomatisation of reversible Kleene lattices} 

\author{Paul Brunet}{University College London, United Kingdom \and \url{paul.brunet-zamansky.fr} }{paul@brunet-zamansky.fr}{https://orcid.org/0000-0002-9762-6872}{}

\authorrunning{P. Brunet}

\Copyright{Paul Brunet}



\category{}

\relatedversion{}

\supplement{Coq formalisation: \url{https://github.com/monstrencage/LangAlg}}

\funding{This work was funded by the ESPRC grant IRIS (reference EP/R006865/1)}

 \nolinenumbers 

 \hideLIPIcs  

\EventEditors{John Q. Open and Joan R. Access}
\EventNoEds{2}
\EventLongTitle{42nd Conference on Very Important Topics (CVIT 2016)}
\EventShortTitle{CVIT 2016}
\EventAcronym{CVIT}
\EventYear{2016}
\EventDate{December 24--27, 2016}
\EventLocation{Little Whinging, United Kingdom}
\EventLogo{}
\SeriesVolume{42}
\ArticleNo{23}

\usepackage{mathpartir}
\usepackage{thmtools}
\usepackage{mathtools}
\usepackage{multicol}

\DeclareSymbolFont{stmry}{U}{stmry}{m}{n}
\DeclareMathDelimiter\llbracket{\mathopen}{stmry}{"4A}{stmry}{"71}
\DeclareMathDelimiter\rrbracket{\mathclose}{stmry}{"4B}{stmry}{"79}

\DeclareMathAlphabet{\mathpzc}{OT1}{pzc}{m}{it}

\newcommand\Mid{\mathrel{|}}
\newcommand\Coloneqq{\mathrel{\mathop{::}}=}
\newcommand\paren[1]{\left(#1\right)}
\newcommand\set[1]{\left\{#1\right\}}
\newcommand\setcompr[2]{\left\{#1~\middle|~#2\right\}}
\newcommand\tuple[1]{\left\langle#1\right\rangle}
\newcommand\eqdef{\mathrel{\mathop{:}}=}
\newcommand\pset[1]{\mathcal{P}\paren{#1}}
\newcommand\fpset[1]{\mathcal{P}_f\paren{#1}}
\newcommand\Nat{\mathbb{N}}

\renewcommand\epsilon\varepsilon
\newcommand\Lang[1]{\mathcal{L}\tuple{#1}}
\newcommand\mirror[1]{\overline{#1}}
\newcommand\sem[2][\sigma]{\left\llbracket {#2}\right\rrbracket_{#1}}

\newcommand\Eklc[1]{\mathbb E_{#1}}
\newcommand\Eklcm[1]{\mathbb E'_{#1}}
\newcommand\Etop[1]{\mathbb E^\top_{#1}}
\newcommand\clean[1]{\mathbb C_{#1}}
\newcommand\Eklm[1]{\mathbb E^-_{#1}}

\newcommand\term[1][X]{\mathbb{T}_{#1}}

\newcommand\valid{\mathbb V}
\newcommand\NF{\mathbb {NF}}
\newcommand\nf[1]{\mathcal {N}\paren{#1}}
\newcommand\interone{\mathbb{I}}

\newcommand\test[1]{\theta_{#1}}
\newcommand\erase[1]{\phi\paren{#1}}
\newcommand\reduce[2]{\tuple{#2}_{#1}}

\newcommand\conv[1]{\overline{#1}}
\newcommand\argument{\square}
\newcommand\comb{\Upsilon}
 
\newcounter{mysubtable}
\newcounter{myaxiom}
\newcommand\axlabel[1]{%
  \stepcounter{myaxiom}%
  \label{#1}%
  \tag{\thetable\alph{mysubtable}.{\scriptsize{\arabic{myaxiom}}}}%
}
\newenvironment{tableequations}{%
  \setcounter{myaxiom}{0}%
  \setcounter{mysubtable}{\value{subtable}}%
  \stepcounter{mysubtable}%
  \ignorespaces
}{%
  \ignorespacesafterend
}
\numberwithin{equation}{section}

\begin{document}

\maketitle

\begin{abstract}
  We consider algebras of languages over the signature of reversible Kleene lattices, that is the regular operations (empty and unit languages, union, concatenation and Kleene star) together with intersection and mirror image.
  We provide a complete set of axioms for the equational theory of these algebras.
  This proof was developed in the proof assistant Coq.
\end{abstract}

\section{Introduction}
We are interested in algebras of languages, equipped with the constants empty language~($0$), unit language~($1$, the language containing only the empty word), the binary operations of union~($+$), intersection~($\cap$), and concatenation~($\cdot$), and the unary operations of Kleene star~($\argument^\star$) and mirror image~($\conv\argument$).
It is convenient in this paper to see the Kleene star as a derived operator $e^\star\eqdef 1+e^+$ with the operator $e^+$ representing the non-zero iteration.
We call these algebras \emph{reversible Kleene lattices}.
Given a finite set of variables~$X$, and two terms~$e,f$ built from variables and the above operations, we say that the equation~$e\simeq f$ is \emph{valid} if the corresponding equality holds universally.

In a previous paper~\cite{brunetReversibleKleeneLattices2017a} we have presented an algorithm to test the validity of such equations, and shown this problem to be \textsc{ExpSpace}-complete.
However, we had left open the question of the axiomatisation of these algebras.
We address it now, by providing in the current paper a set of axioms from which every valid equation can be derived.

Several fragments of this algebra have been studied:
\begin{description}
\item[Kleene algebra (KA):]
  if we restricts ourselves to the operators of regular expressions ($0$, $1$, $+$, $\cdot$, and $\argument^+$), then several axiomatisation have been proposed by Conway\cite{Conway}, before being shown to be complete by Krob~\cite{krobCompleteSystemsBrational1991} and Kozen~\cite{kozenCompletenessTheoremKleene1994}.
\item[Kleene algebra with converse:]
  if we add to KA the mirror operation, then the previous theorem can be extended by switching to a duplicated alphabet, with a letter $a'$ denoting the mirror of the letter $a$.
  A small number of identities may be added to KA to get a complete axiomatisation~\cite{BES95}.
\item[Identity-free Kleene lattices:]
  this algebra stems from the operators $0$, $+$, $\cdot$, $\cap$ and $\argument^+$.
  In a recent paper~\cite{doumaneCompletenessIdentityfreeKleene2018a} Doumane and Pous provided a complete axiomatisation of this algebra.
\end{description}

The present work is then an extension of identity-free Kleene lattices, by adding unit and mirror image.
We provide in \Cref{tab:axioms} a set of axioms which we prove to be complete for the equational theory of language algebra, by reducing to the the completeness theorem of~\cite{doumaneCompletenessIdentityfreeKleene2018a}. This proof has been formalised in \href{http://coq.inria.fr}{Coq}.

The paper is organised as follows.
In \Cref{sec:prelim}, we introduce some notations and define the various types of expressions used in the paper.
We present our axioms and state our main theorem.
In \Cref{sec:rm-nil} we deal with a technical lemma having to do with the treatment of the empty word.
We proceed in \Cref{sec:mirror} to extend the theorem of~\cite{doumaneCompletenessIdentityfreeKleene2018a} with the mirror image operator.
\Cref{sec:tests} studies in details terms of the algebra that are below the constant $1$, as those play a crucial role in the main proof.
We present the proof of our main result in \Cref{sec:rkl}.
We conclude in \Cref{sec:top} by a discussion on an operator that is missing from our signature, namely constant $\top$ denoting the full language.

\section{Preliminaries}
\label{sec:prelim}

\subsection{Sets, words, and languages}
\label{sec:prelim:words}
Given a set $X$, we write $\pset X$ for the set of subsets of $X$ and $\fpset X$ for the set of finite subsets of $X$.
We will denote the two-elements boolean set $2$.
For two sets $X,Y$, we write $X\times Y$ for their Cartesian product, $X\cup Y$ for their union, and $X\cap Y$ for their intersection.
The empty set is denoted by $\emptyset$.
We will use the notation $f(A)$ for a set $A\subseteq X$ and a function $f:X\to Y$ to represent the set $\setcompr{y\in Y}{\exists a\in A:f(a) = y}=\setcompr{f(a)}{a\in A}$.

Let $\Sigma$ be an arbitrary alphabet (set), the words over $\Sigma$ are finite sequences of elements from $\Sigma$.
The set of all words is written $\Sigma^\star$, and the empty word is written $\epsilon$.
The concatenation of two words $u,v$ is simply denoted by $uv$.
The mirror image of a word $u$, obtained by reading it backwards, is written $\mirror u$.
For instance $\mirror{abc}$ is the word $cba$.

A language is a set of words, that is an element of $\Lang\Sigma\eqdef\pset{\Sigma^\star}$.
We will also use the symbol $\epsilon$ to denote the unit language $\set\epsilon$.
The concatenation of two languages $L$ and $M$, denoted by $L\cdot M$, is obtained by lifting pairwise the concatenation of words: it contains exactly those words that can be obtained as a concatenation $uv$ where $\tuple{u,v}\in L\times M$.
Similarly the mirror image of a language $L$, denoted by $\mirror L$, is the set of mirror images of words from $L$.
We write $L^n$ when $L\in\Lang\Sigma$ and $n\in\Nat$ for the iterated concatenation, defined by induction on $n$ by $L^0 \eqdef \epsilon$ and $L^{n+1}\eqdef L\cdot L^n$.
The language $L^+$ is the union of all non-zero iterations of $L$, i.e. $L^+\eqdef \bigcup_{n>0}L^n$.

\subsection{Terms: syntax and semantics}
\label{sec:prelim:terms}

Throughout this paper, we will consider expressions over various signatures which we list here.
We fix a set of variables $X$, and let $x,y,...$ range over $X$.
\begin{description}
\item[Expressions:] $e,f\in \Eklc X\Coloneqq x \Mid 0\Mid 1 \Mid e + f \Mid e \cdot f \Mid e \cap f \Mid e^+ \Mid \conv{e}$;
\item[One-free expressions:] $e,f\in \Eklcm X\Coloneqq x \Mid 0 \Mid e + f \Mid e \cdot f \Mid e \cap f \Mid e^+ \Mid \conv{e}$;
\item[Simple expressions:] $e,f\in \Eklm X\Coloneqq x \Mid 0 \Mid e + f \Mid e \cdot f \Mid e \cap f \Mid e^+$;
\end{description}

\begin{table}
  \centering
  \begin{subtable}[b]{.405\textwidth}
    \begin{tableequations}
      \begin{align}
        e+f &= f+e\axlabel{ax:plus-com}\\
        e+(f+g) &= (e+f)+g \axlabel{ax:plus-ass}\\
        e+0 &= e  \axlabel{ax:plus-0}\\
        e\cap f  &=f\cap e\axlabel{ax:inter-comm}\\
        e \cap  e &= e\axlabel{ax:inter-idem}\\
        e\cap (f \cap  g) &= (e\cap f)\cap g\axlabel{ax:inter-assoc}\\
        (e + f)\cap g &= e\cap g + f\cap g\axlabel{ax:plus-inter}\\
        (e\cap f)+e &= e\axlabel{ax:inter-plus}
      \end{align}
    \end{tableequations}
    \caption{Distributive lattice}
    \label{tab:lattice}
  \end{subtable}\hfill%
  \begin{subtable}[b]{.58\textwidth}
    \begin{tableequations}
      \begin{align}
        e\cdot (f \cdot  g) &= (e\cdot f)\cdot g\axlabel{ax:seq-assoc}\\
        e\cdot 0 &= 0= 0\cdot e \axlabel{ax:seq-0}\\
        (e + f)\cdot g &= e\cdot g + f\cdot g\axlabel{ax:plus-seq}\\
        e\cdot (f + g) &= e\cdot f + e\cdot g\axlabel{ax:seq-plus}\\
        e^+ &= e + e\cdot e^+\axlabel{ax:iter-left}\\
        e^+ &= e + e^+ \cdot e\axlabel{ax:iter-right}\\
        e\cdot f + f = f &\Rightarrow e^+\cdot f + f = f\axlabel{ax:left-ind}\\
        f \cdot  e + f = f &\Rightarrow f \cdot e^+ + f = f\axlabel{ax:right-ind}
      \end{align}
    \end{tableequations}   \caption{Concatenation and iteration}
    \label{tab:concat}
  \end{subtable}
  \begin{subtable}[b]{.405\textwidth}
    \begin{tableequations}
      \begin{align}
        \conv{{\conv{e }}} &= e\axlabel{ax:conv-conv}\\
        \conv{{e + f}} &= \conv{e } + \conv{f }\axlabel{ax:conv-plus}\\
        \conv{{e \cdot  f}} &= \conv{f } \cdot  \conv{e }\axlabel{ax:conv-seq}\\
        \conv{{e\cap f}} &= \conv{e } \cap  \conv{f }\axlabel{ax:conv-inter}\\
        \conv{{e^+}} &= {\conv{e }}^+\axlabel{ax:conv-iter}
      \end{align}
    \end{tableequations}
    \caption{Mirror image}
    \label{tab:mirror}
  \end{subtable}\hfill%
  \begin{subtable}[b]{.58\textwidth}
    \begin{tableequations}
      \begin{align}
        1\cdot e &= e=e\cdot 1\axlabel{ax:seq-1}\\
        1 \cap  (e\cdot f) &= 1 \cap  (e \cap  f)\axlabel{ax:test-seq-inter}\\
        1 \cap  \conv{e } &= 1 \cap  e\axlabel{ax:test-conv}\\
        (1 \cap  e)\cdot f &= f\cdot (1 \cap  e)\axlabel{ax:test-seq-com}\\
        ((1 \cap  e)\cdot f) \cap  g &= (1 \cap  e)\cdot (f \cap  g)\axlabel{ax:test-inter}\\
        (g + (1 \cap  e) \cdot  f)^+ &= g^+ + (1 \cap  e) \cdot  (g + f)^+\axlabel{ax:test-iter}
      \end{align}
    \end{tableequations}
    \caption{Unit}
    \label{tab:subunits}
  \end{subtable}
  \caption{Axioms of reversible Kleene lattices}
  \label{tab:axioms}
\end{table}
\begin{table}
  \begin{subtable}[b]{.405\linewidth}
    \begin{tableequations}
      \begin{align}
        e+e&\equiv e\axlabel{eq:7}\\
        e\cap 0&\equiv 0\axlabel{eq:8}\\
        e\cap(e+f)&\equiv e\axlabel{eq:9}
      \end{align}
    \end{tableequations}
    \caption{Lattice laws}
  \end{subtable}\hfill%
  \begin{subtable}[b]{.58\linewidth}
    \begin{tableequations}
      \begin{align}
        e^+\cdot e^+&\leqq e^+\axlabel{eq:10}\\
        \paren{e^+}^+&\equiv e^+\axlabel{eq:11}\\
        \paren{1+e}^+&\equiv 1+e^+\axlabel{eq:12}
      \end{align}
    \end{tableequations}
    \caption{Iteration}
  \end{subtable}
  \begin{subtable}[b]{.405\linewidth}
    \begin{tableequations}
      \begin{align}
        \conv 0&\equiv 0\axlabel{eq:13}\\
        \conv 1&\equiv 1\axlabel{eq:14}\\
        0^+&\equiv 0\axlabel{eq:15}\\
        1^+&\equiv 1\axlabel{eq:16}
      \end{align}
    \end{tableequations}
    \caption{Constants}
  \end{subtable}\hfill%
  \begin{subtable}[b]{.58\linewidth}
    \begin{tableequations}
      \begin{align}
        e \leqq g \Rightarrow f \leqq g &\Rightarrow e + f \leqq g\axlabel{eq:17}\\
        g \leqq e \Rightarrow g \leqq e &\Rightarrow g \leqq e\cap f\axlabel{eq:18}\\
        e\leqq f &\Leftrightarrow e\cap f \equiv e\axlabel{eq:19}\\
        1\leqq e\cdot f&\Leftrightarrow 1\leqq e~\wedge~ 1\leqq f\axlabel{eq:20}
      \end{align}
    \end{tableequations}
    \caption{Reasoning rules}
  \end{subtable}
  \caption{Some consequences of the axioms}
  \label{tab:derivable}
\end{table}
We will use various sets of axioms, depending the signature.
All of the axioms under consideration are listed in~\Cref{tab:axioms}.
Axioms in \Cref{tab:lattice,tab:concat} are borrowed from~\cite{doumaneCompletenessIdentityfreeKleene2018a}, those in~\Cref{tab:mirror} from \cite{BES95} and \Cref{tab:subunits} is inspired from \cite{andrekaEquationalTheoryKleene2011}.
We use these axioms to generate equivalence relations over terms.
For a type of expressions $\term\in\set{\Eklc X,\Eklcm X,\Eklm X}$, the \emph{axiomatic equivalence relation}, written $\equiv$ is the smallest congruence on $\term$ containing those axioms in~\Cref{tab:axioms} that only use symbols from the signature of $\term$.
This means that for $\Eklm X$ we use the axioms from \Cref{tab:lattice,tab:concat}, for $\Eklcm X$ we add those from \Cref{tab:mirror} and for $\Eklc X$ we keep all of the axioms of \Cref{tab:axioms}.
We will use the shorthand $e\leqq f$ to mean $e+f\equiv f$.
This ensures that $\leqq$ is a partial order with respect to $\equiv$.
We list in \Cref{tab:derivable} some statements that are provable from the axioms.
Interestingly the idempotency of $+$ (equation \eqref{eq:7}), which is usually an axiom, is here derivable from \eqref{ax:inter-idem} and~\eqref{ax:inter-plus}.

Given an expression $e\in\term$, a set $\Sigma$, and a map $\sigma:X\to\Lang\Sigma$, we may interpret $e$ as a language over $\Sigma$ using the following inductive definition:
\begin{align*}
  \sem x &\eqdef \sigma(x)
  &\sem {e+f} &\eqdef \sem e \cup \sem f
  &\sem {e^+} &\eqdef \sem e^+
  \\
  \sem 0 &\eqdef \emptyset
  &\sem {e\cdot f} &\eqdef \sem e \cdot \sem f
  &\sem {\conv{e}} &\eqdef \mirror{\sem e}
  \\
  \sem 1 &\eqdef \epsilon
  &\sem {e\cap f} &\eqdef \sem e \cap \sem f
\end{align*}
The \emph{semantic equivalence} and \emph{semantic containment} relations on $\term$, respectively written $\simeq$ and $\lesssim$, are defined as follows:
\begin{align*}
e\simeq f &\Leftrightarrow \forall \Sigma,\,\forall \sigma:X\to\Lang\Sigma,\,\sem e = \sem f.\\
e\lesssim f &\Leftrightarrow \forall \Sigma,\,\forall \sigma:X\to\Lang\Sigma,\,\sem e \subseteq \sem f.
\end{align*}

The main result of this paper is a completeness theorem for reversible Kleene lattices:
\begin{restatable*}[Main result]{theorem}{complrkl}\label{thm:compl-rkl}
  $\forall e,f\in\Eklc X,\; e\equiv f \Leftrightarrow e\simeq f$.
\end{restatable*}
Since all of the axioms in~\Cref{tab:axioms} are sound for languages, we know that the implication from left to right holds.
This paper will thus focus on the converse implication, and will proceed in several steps.
Our starting point will be the recently published completeness theorem for identity-free Kleene lattices~\cite{doumaneCompletenessIdentityfreeKleene2018a}:
\begin{theorem}\label{thm:compl-klm}
  $\forall e,f\in\Eklm X,\; e\equiv f \Leftrightarrow e\simeq f$.%
\end{theorem}
\begin{remark}
  In~\cite{doumaneCompletenessIdentityfreeKleene2018a}, this theorem is established for interpretations of terms as binary relations instead of languages.
  However both semantic equivalences coincide for this signature~\cite{andrekaEquationalTheoryKleene2011}.
\end{remark}
\section{A remark about the empty word}
\label{sec:rm-nil}

In several places in the proof, it makes some difference whether or not the empty word belongs to the language of some one-free expression.
We show here one way one might manipulate this property, that will be of use later on.
\begin{restatable}{proposition}{rmnil}\label{prop:rm-nil}
  Given an alphabet $\Sigma$, a map $\sigma:X\to\Lang\Sigma$ and a set of variables $\mathcal X\subseteq X$, there is an alphabet $\Sigma'$, and two maps $\sigma':X\to\Lang{\Sigma'}$ and $\erase\argument:\Sigma'^\star\to\Sigma^\star$ such that:
  \begin{align*}
    \forall a\in \mathcal X,\;\epsilon\notin\sigma'(a)&&
    \forall e\in\Eklcm{X},\;\sem e =\erase{\sem[\sigma']e\setminus\epsilon}.
    \end{align*}
\end{restatable}
\begin{proof}
  We fix $\Sigma$, $\sigma$, and $\mathcal X$ as in the statement.
  Let $\bullet$ be some fresh letter, we set $\Sigma'$ to be $\Sigma\cup\set\bullet$.
  For a word $u\in\Sigma'^\star$, we define $\erase u\in\Sigma^\star$ to be the word obtained by removing every instance of $\bullet$ from $u$.
  Finally, $\sigma'$ is defined as follows:
  \[
    \sigma'(x)\eqdef\setcompr{u}{\erase u\in \sigma(x) \wedge \paren{x\in \mathcal X\Rightarrow u\neq \epsilon}}.
  \]
  It is straightforward to check that $\erase{\sigma'(x)}=\sigma(x)$ for any variable $x$.
  Therefore we only need to check that this property is preserved by the operators of one-free expressions.
  For any languages $L,M$, the following distributivity laws hold:
  \begin{align*}
    \erase{\conv{L}}&=\conv{\erase{L}}
    &\erase{L\cdot M}&=\erase L\cdot\erase M\\
    \erase{L^+}&=\erase L^+
    &\erase{L\cup M}&=\erase L \cup\erase M
  \end{align*}
  However, it is not the case in general that $\erase{L\cap M}=\erase L\cap\erase M$.
  To make the induction go through, we will need to show that this identity holds for all the languages generated from the languages $\sigma'(x)$ by the operations $0,\cdot,+,\cap,\argument^+,\conv\argument$.
  This is achieved by identifying some sufficient condition for $\erase{L\cap M}=\erase L\cap\erase M$, and showing that this condition is satisfied by every language  of the shape $\sem[\sigma']e$.
  Let us define the ordering $\sqsubseteq$ on words over $\Sigma'$:
  \begin{align*}
    \infer{}{u\sqsubseteq u}&&
    \infer{}{u\sqsubseteq \bullet u}&&
    \infer{u\sqsubseteq v\and v\sqsubseteq w}{u\sqsubseteq w}&&
    \infer{u\sqsubseteq v\and u'\sqsubseteq v'}{uu'\sqsubseteq vv'}
    \end{align*}
  $\sqsubseteq$ is a partial order and satisfies the following properties:
  \begin{align}
    &u\sqsubseteq v \Rightarrow \erase u = \erase v\label{eq:1}\\
    &u\sqsubseteq v \Rightarrow \mirror u\sqsubseteq \mirror v\label{eq:2}\\
    &u_1u_2\sqsubseteq v \Rightarrow \exists v_1,v_2:v= v_1v_2\wedge u_1\sqsubseteq v_1\wedge u_2\sqsubseteq v_2\label{eq:3}\\
    &\erase u = \erase v\Rightarrow\exists w, u\sqsubseteq w\wedge v\sqsubseteq w\wedge\paren{\forall w', u\sqsubseteq w'\wedge v\sqsubseteq w'\Rightarrow w\sqsubseteq w'}.\label{eq:4}
  \end{align}
  \eqref{eq:1} and \eqref{eq:4} tell us that each equivalence class of the relation $\setcompr{\tuple{u,v}}{\erase u=\erase v}$ forms a join-semilattice.
  The proofs of these properties being somewhat technical, we omit them here.
  The interested reader may refer to the Coq formalisation for details.

  Consider now those languages over $\Sigma'$ that are upwards-closed with respect to $\sqsubseteq$, that is to say languages $L$ such that whenever $u\in L$ and $u\sqsubseteq v$, then $v\in L$.
  Clearly $\sigma'(x)$ is closed for any variable $x$.
  Since the property ``being closed'' is preserved by each operation in the signature of $\Eklcm X$, we deduce that for any expression $e\in\Eklcm X$ the language $\sem[\sigma']e$ is closed.
  
  Thankfully, for closed languages the missing identity $\erase{L\cap M}=\erase L\cap\erase M$ holds.
  Thus we may conclude by induction on the expressions that $\sem e=\erase{\sem[\sigma'] e}$.
  For the last step, notice that $\epsilon\sqsubseteq\bullet$ and $\erase\epsilon=\erase\bullet$.
  Since $\sem[\sigma']e$ is closed, if $\epsilon\in\sem[\sigma']e$ then $\bullet\in\sem[\sigma']e$, thus $\erase{\sem[\sigma']e\setminus\epsilon}=\erase{\sem[\sigma']e}=\sem e$.
\end{proof}

By setting the set $\mathcal X$ in the previous proposition to the full set $X$, we get the straightforward corollary, which will prove useful in the next section.
\begin{corollary}\label{cor:rm-nil}
  Let $e$ be a one-free expression, then for any expression $f\in\Eklc X$ we have
  \[e\lesssim f \Leftrightarrow\forall \Sigma,\,\forall \sigma:X\to\Lang\Sigma,\,\epsilon\notin\bigcup_{x\in X}\sigma(x) \Rightarrow \sem e \subseteq \sem f.\]
\end{corollary}

\section{Mirror image}
\label{sec:mirror}
In this section, we show a completeness theorem for one-free expressions.
In order to get this result we will use translations between $\Eklcm X$ and $\Eklm {X\times 2}$.
An expression $e\in\Eklcm X$ is \emph{clean}, written $e\in\clean X$, if the mirror operator is only applied to variables.
First, notice that we may restrict ourselves to clean expressions thanks to the following inductive function:
\begin{align*}
  \comb :\Eklcm X \times 2 &\to \Eklcm X
  &&\\
  \tuple{0,b}&\mapsto 0
  &\tuple{e^+,b}&\mapsto \comb\tuple{e,b}^+\\
  \tuple{x,\top}&\mapsto x
  &\tuple{\conv e,\top}&\mapsto \comb\tuple{e,\bot}\\
  \tuple{x,\bot}&\mapsto \conv{x}
  &\tuple{\conv e,\bot}&\mapsto \comb\tuple{e,\top}\\
  \tuple{e+f,b}&\mapsto \comb\tuple{e,b}+\comb\tuple{f,b}
  &  \tuple{e\cdot f,\top}&\mapsto \comb\tuple{e,\top}\cdot \comb\tuple{f,\top}\\
  \tuple{e\cap f,b}&\mapsto \comb\tuple{e,b}\cap \comb\tuple{f,b}
  & \tuple{e\cdot f,\bot}&\mapsto \comb\tuple{f,\bot}\cdot \comb\tuple{e,\bot}.
\end{align*}
We can show by induction on terms the following properties of $\comb$:
\begin{align}
  \forall \tuple{e,b}\in\Eklcm X\times 2,&\;\comb\tuple{e,b}\in\clean X.\label{lem:comb-clean}\\
  \forall e\in\Eklcm X,&\;\comb\tuple{e,\top}\equiv e\text{ and } \comb\tuple{e,\bot}\equiv \conv{e}.\label{lem:comb-ok}
\end{align}
We now define translations between clean expressions and simple expressions:
\begin{itemize}
\item $\uparrow:\clean X\to\Eklm{X\times 2}$ replaces mirrored variables $\conv{x}$ with $\tuple{x,\bot}$ and variables $x$ with $\tuple{x,\top}$;
\item $\downarrow:\Eklm{X\times 2}\to\clean X$ replaces $\tuple{x,\top}$ with $x$ and $\tuple{x,\bot}$ with $\conv{x}$.
\end{itemize}
We can easily show by induction the following properties:
\begin{align}
  \forall e\in\clean X,&\,\downarrow\uparrow e = e.\label{lem:down-up}\\
  \forall e,f \in \Eklm{X\times 2},&\,e \equiv f \Rightarrow \downarrow e \equiv \downarrow f.\label{lem:down-eq}
\end{align}

The last step to obtain the completeness theorem for $\Eklcm X$ is the following claim:
\begin{claim}\label{claim:comb-lang}
  $\forall e,f\in\clean X,\, e\simeq f\Rightarrow\uparrow e \simeq \uparrow f$.
\end{claim}

\begin{lemma}\label{lem:completeness-klcm}
  If \Cref{claim:comb-lang} holds, then
  $\forall e,f\in\Eklcm X,\, e\equiv f \Leftrightarrow e\simeq f$.
\end{lemma}
\begin{proof}
  By soundness, we know that $e\equiv f \Rightarrow e\simeq f$.
  For the converse implication:
  \begin{align*}
    e\simeq f
    &\Rightarrow \comb\tuple{e,\top}\simeq \comb\tuple{f,\top}
      \tag*{By soundness and  \Cref{lem:comb-ok}.}\\
    &\Rightarrow\uparrow \comb\tuple{e,\top} \simeq \uparrow \comb\tuple{f,\top}
      \tag*{By \Cref{claim:comb-lang}.}\\
    &\Rightarrow\uparrow \comb\tuple{e,\top} \equiv \uparrow \comb\tuple{f,\top}
      \tag*{By \Cref{thm:compl-klm}.}\\
    &\Rightarrow\downarrow\uparrow \comb\tuple{e,\top} \equiv \downarrow\uparrow \comb\tuple{f,\top}
      \tag*{By \Cref{lem:down-eq}.}\\
    &\Rightarrow \comb\tuple{e,\top} \equiv \comb\tuple{f,\top}
      \tag*{By \Cref{lem:down-up}.}\\
    &\Rightarrow e \equiv f
      \tag*{By \Cref{lem:comb-ok}.}
  \end{align*}
\end{proof}

Hence, we only need to show \Cref{claim:comb-lang} to conclude.
To that end, we show that for any clean expression $e$, any interpretation of $\uparrow e$ can be obtained by applying some transformation to some interpretation of $e$.
Thanks to \Cref{cor:rm-nil}, we may restrict our attention to interpretation that avoid the empty word. This seemingly mundane restriction turns out to be of significant importance: if the empty word is allowed, the proof of \Cref{lem:rm-conv} becomes much more involved.
More precisely, we prove the following lemma:

\begin{lemma}\label{lem:rm-conv}
  Let $\Sigma$ be some set and $\sigma:X\times 2\to\Lang \Sigma$ some interpretation such that $\forall x, \epsilon\notin \sigma(x)$.
  There exists an alphabet $\Sigma'$, an interpretation $\sigma'':X\to\Lang{\Sigma'}$ and a function $\psi:\Lang{\Sigma'}\to\Lang{\Sigma}$ such that:
  $\forall e\in\clean X,\;\sem{\uparrow e}=\psi\paren{\sem[\sigma'']e}$.
\end{lemma}
\begin{proof}
  We fix $\Sigma$ and $\sigma:X\times 2\to\Lang \Sigma$ as in the statement.
  Like in the proof of \Cref{prop:rm-nil}, we set $\Sigma'=\Sigma\cup\set{\bullet}$, with $\bullet$ a fresh letter, and write $\erase u$ for the word obtained from $u\in\Sigma'^\star$ by erasing every occurrence of $\bullet$.
  Additionally we define the function $\eta:\Sigma^\star\to\Sigma'^\star$ as follows:
  \begin{align*}
    \eta(\epsilon)\eqdef\epsilon&& \eta(a\,u)\eqdef\bullet \,a\,\eta(u)\quad(\tuple{a,u}\in\Sigma\times\Sigma^\star).\end{align*}
  Clearly, $\erase{\eta(u)}=u$ and $\eta\paren{u\,v}=\eta(u)\,\eta(v)$.
  We may now define $\sigma''$ and $\psi$:
  \begin{align*}
    \sigma''(x)\eqdef\setcompr{\eta(u)}{u\in\sigma\tuple{x,\top}}
    \cup\setcompr{\mirror{\eta(u)}}{u\in\sigma\tuple{x,\bot}}&&
    \psi\paren{L}\eqdef\setcompr{u}{\eta(u)\in L}.\end{align*}
  This is where the restriction $\epsilon\notin\sigma(x)$ comes in.
  Indeed a word $w$ cannot be written both as $w=\eta(u_1)$ and as $w=\mirror{\eta(u_2)}$ unless $w=u_1=u_2=\epsilon$.
  Since $\sigma$ does not contain the empty word, we may show that  $\psi\paren{\sigma''(x)}=\sigma\tuple{x,\top}$ and
  $\psi\paren{\conv{\sigma''(x)}}=\sigma\tuple{x,\bot}$.
  
  $\psi$ distributes over the union and intersection operators.
  However, it does not hold in general that $\psi\paren{L\cdot M}=\psi(L)\cdot\psi(M)$.
  Like in the proof of \Cref{prop:rm-nil} we will therefore identify a predicate on languages that is sufficient for this identity to hold, is satisfied by $\sigma''(x)$, and is stable by $\cdot,\cap,+,\argument^+,\conv\argument$.
  In this case we find that an adequate candidate is ``$L$ contains only valid words'', where the set $\valid$ of valid words is defined as follows:
  \begin{align*}
    \infer{u\in\Sigma^+}{\eta(u)\in\valid}&&
    \infer{u\in\valid}{\mirror u\in\valid}&&
    \infer{u\in\valid\and v\in\valid}{u\,v\in\valid}
    \end{align*}
  Alternatively, the elements of $\valid$ are words over $\Sigma'$ that can be written as a product $\alpha_1\dots\alpha_n$ with $1\leqslant n$ and each $\alpha_i\in\paren{\Sigma\cdot\bullet}\cup \paren{\bullet\cdot\Sigma}$.
  One may see from the definitions that $\sigma''(x)\subseteq \valid$.
  $\valid$ can also be seen to be trivially closed by concatenation and mirror image.
  Since the remaining operators are either idempotent (union and intersection) or derived (iteration), we get that $\sem[\sigma'']e\subseteq \valid$.
  This enables us to conclude thanks to the following property:
  \begin{equation}
    \label{eq:valid-app}
    \forall u_1,u_2\in\valid,\; \eta(u) = u_1\, u_2\Rightarrow
    \exists v_1, v_2: u_1 = \eta(v_1) \wedge u_2 = \eta(v_2) \wedge u = v_1\,v_2.
  \end{equation}
  This property enables us to show that $\psi\paren{L\cdot M}=\psi(L)\cdot\psi(M)$ and $\psi\paren{L^+}=\psi\paren L^+$, for languages of valid words $L,M$.
  Hence we obtain by induction on expressions that for any term $e\in\clean X$, it holds that $\sem[\sigma]{\uparrow e}=\sem[\sigma'']{e}$.
\end{proof}
\begin{theorem}\label{thm:completeness-klcm}
  $\forall e,f\in\Eklcm X,\, e\equiv f \Leftrightarrow e\simeq f$.
\end{theorem}
\begin{proof}
  Thanks to \Cref{lem:completeness-klcm}, we only need to check \Cref{claim:comb-lang}.
  Let $e,f$ be two clean expressions such that $e\simeq f$, we want to prove $\uparrow e\simeq \uparrow f$.
  According to \Cref{cor:rm-nil}, we need to compare $\sem{\uparrow e}$ and $\sem{\uparrow f}$ for some $\sigma:X\times 2\to\Lang\Sigma$ such that $\epsilon\notin\bigcup_{x\in X\times 2}\sigma(x)$.
  By \Cref{lem:rm-conv}, we may express these languages as respectively $\psi\paren{\sem[\sigma'']{ e}}$ and $\psi\paren{\sem[\sigma'']{ f}}$.
  Since $e\simeq f$, we get that $\sem[\sigma'']{ e}=\sem[\sigma'']{f}$, thus proving the desired identity and concluding the proof.
\end{proof}
\section{Interlude: tests}\label{sec:tests}
Before we start with the main proof, we define \emph{tests} and establish a few result about them. 
Given a list of variables $u\in X^\star$, we define the term $\test u$ by induction on $u$ as $\test \epsilon \eqdef 1$ and $\test {a\, u}\eqdef a\cap\test u$.
Thanks to the following remark, we will hereafter consider $\test A$ for $A\in\fpset X$: 
\begin{remark}\label{rmk:test-set}
  Let $u,v$ be two lists of variables containing the same letters (meaning a variable appears in $u$ if and only if it appears in $v$). Then $\test u\equiv \test v$.
\end{remark}

The following property explains out choice of terminology: the function $\lambda\sigma.\sem{\test A}$ can be seen as a boolean predicate testing whether the empty word is in each of the $\sigma(a)$ for $a\in A$.
\begin{lemma}\label{lem:sem-test}
  Let $\Sigma$ be some alphabet and $\sigma:X\to\Lang\Sigma$. Then either $\forall a\in A,\;\epsilon\in \sigma(a)$, in which case $\sem{\test A}=\epsilon$, or $\sem{\test A}=\emptyset$.
\end{lemma}
\noindent%
Tests satisfy the following universal identities, with $A,B\in\fpset X$ and $e,f\in\Eklc X$:
\begin{align}
  &\test A\leqq 1\label{eq:test-sub-id}\\
  &\test A\cap \test B\equiv\test A\cdot \test B\equiv\test {A\cup B}\\
  &\test A \equiv\test A\cdot \test A\label{eq:5}\\
  &a\in A\Rightarrow\test A \leqq a\label{eq:6}\\
  &\test A\cdot e\equiv e\cdot \test A \\
  &(\test A \cdot e)\cap(\test B \cdot f)\equiv\test {A\cup B} \cdot (e \cap f)\\
  &\test A ^+\equiv \conv{\test A}\equiv \test A.
\end{align}
We now want to compare tests with other tests or with expressions.
Let us define the following interpretation for any finite set $A\in\fpset X$.
\[\begin{array}[t]{crl}
  \sigma_A:&X&\to\Lang\emptyset\footnotemark\\
           &x&\mapsto\left\{
               \begin{array}{ll}
                 \epsilon&\text{if}\;x\in A\\
                 \emptyset&\text{otherwise}.
               \end{array}
                            \right.
\end{array}
\]
\footnotetext{The alphabet here does not matter, since we only want the unit language and the empty language.}
This enables us to establish the following lemma:
\begin{lemma}\label{lem:tests-inf}
  For any $A,B\in\fpset X$, the following are equivalent:
  \begin{multicols}{4}
    \begin{enumerate}[(i)]
    \item\label{item:0} $\epsilon\in\sem[\sigma_A]{\test B}$
    \item\label{item:1} $B\subseteq A$
    \item\label{item:2} $\test A\leqq \test B$
    \item\label{item:3} $\test A\lesssim \test B$.
    \end{enumerate}
  \end{multicols}
\end{lemma}
\begin{proof}
  Assume \eqref{item:0} holds, i.e. $\epsilon\in\sem[\sigma_A]{\test B}$.
  By \Cref{lem:sem-test} this means that for every $a\in B$ we have $\epsilon\in\sigma_A(a)$ which by definition of $\sigma_A$ ensures that $a\in A$.
  Thus we have shown that \eqref{item:1} holds.
  We show that \eqref{item:1} implies \eqref{item:2} by induction on the size of $B$:
  \begin{itemize}
  \item if $B=\emptyset$, by \Cref{eq:test-sub-id} $\test A\leqq 1=\test\emptyset$.
  \item if $B=\set a\cup B'$ with $a\notin B'$, since $B\subseteq A$ we have $a\in A$ and $B'\subseteq A$.
    By induction hypothesis we know that $\test A\leqq \test B'$.
    By \Cref{rmk:test-set} we get that $\test A\equiv a\cap\test A$.
    Hence we get:
    \[\test A\equiv a \cap \test A\leqq a\cap \test B' = \test B.\]
  \end{itemize}
  Thanks to soundness we have that \eqref{item:2} implies \eqref{item:3}.
  For the last implication, notice that by construction of $\sigma_A$ we have $\epsilon\in\sem[\sigma_A]{\test A}$.
  Therefore if $\test A\lesssim\test B$ then we can conclude that $\epsilon\in\sem[\sigma_A]{\test A}\subseteq\sem[\sigma_A]{\test B}$.
\end{proof}

We now define a function $\interone:\Eklc X\to\fpset{\fpset X}$, whose purpose is to represent as a sum of tests the intersection of an arbitrary expression with $1$:
\begin{align*}
  \interone(0)\eqdef\emptyset&&
  \interone(1)\eqdef\set{\emptyset}&&
  \interone(x)\eqdef\set{\set x}&&
  \interone(e+f)\eqdef \interone(e)\cup \interone(f)&&
  \end{align*}
\begin{align*}
  \interone(e\cdot f)=\interone(e\cap f)\eqdef
  \setcompr{A\cup B}{\tuple{A,B}\in \interone(e)\times \interone(f)}&&
  \interone(e^+)=\interone(\conv{e})\eqdef \interone(e).
\end{align*}

\begin{lemma}\label{lem:interone}
  $\forall e\in\Eklc X,\;1\cap e\equiv\sum_{C\in \interone(e)}\test C$.
\end{lemma}

\begin{corollary}\label{lem:tests-compl}
  $\forall e\in\Eklc X,\forall A\in\fpset X,\; \test A\leqq e \Leftrightarrow \test A\lesssim e$.
\end{corollary}
\begin{proof}
  We only need to show the implication from right to left.
  Assume $\test A\lesssim e$.
  This implies $1\cap\test A\lesssim 1\cap e$, and since $\test A\leqq 1$ we know that $1\cap\test A\equiv \test A$ which by soundness implies $\test A\simeq 1\cap \test A$.
  Combining this with \Cref{lem:interone}, we get that
  $\test A\simeq 1\cap \test A\lesssim 1\cap e\simeq\sum_{C\in \interone(e)}\test C$.
  By \Cref{lem:tests-inf}, we know that $\epsilon\in\sem[\sigma_A]{\test A}$, which means that $\epsilon\in\sem[\sigma_A]{\sum_{C\in \interone(e)}\test C}=\bigcup_{C\in \interone(e)}\sem[\sigma_A]{\test C}$.
  Therefore there must be some $B\in\interone(e)$ such that $\epsilon\in\sem[\sigma_A]{\test B}$ which by \Cref{lem:tests-inf} tells us that $\test A\leqq \test B$.
  We may now conclude:
  \[\test A\leqq \test B\leqq\sum_{C\in \interone(e)}\test C\equiv 1\cap e\leqq e.\qedhere\]
\end{proof}

\begin{remark}
  The word ``test'' is reminiscent of Kleene algebra with tests (KAT)\cite{kozenKleeneAlgebraTests1997}.
  Indeed according to \Cref{eq:test-sub-id} our tests are sub-units, like in KAT.
  However unlike in KAT, not every sub-unit is a test.
  Instead here sub-units are in general sums of tests, as can be inferred from \Cref{lem:interone} (because for every sub-unit $e\leqq 1$, we have $e\equiv 1\cap e\equiv\sum_{C\in\interone(e)}\test C$).
\end{remark}
\section{Completeness of reversible Kleene lattices}
\label{sec:rkl}

To tackle this completeness proof, we will proceed in three steps.
Since we already know soundness, and since an equality can be equivalently expressed as a pair of containments, we start from the following statement:
\[\forall e,f\in\Eklc X,\; e\lesssim f \Rightarrow e\leqq f.\]
First, we will show that any expression in $\Eklc X$ can be equivalently written as a sum of terms that are either tests or products $\test A\cdot e$ of a test and a one-free expression.
The case of tests having been dispatched already (\Cref{lem:tests-compl}), this reduces the problem to:
\[\forall e\in\Eklcm X,\,\forall A\in\fpset X,\,\forall f\in\Eklc X,\; \test A\cdot e\lesssim f \Rightarrow \test A\cdot e\leqq f.\]
Second, we will show that for any pair $\tuple{A,f}\in \fpset X\times\Eklc X$, there exists an expression $\reduce A f\in\Eklc X$ such that $\test A\cdot \reduce A f\leqq f$ and whenever $\test A\cdot e\lesssim f$ we have $e\lesssim \reduce A f$.
This further reduces the problem into:
\[\forall e\in\Eklcm X,\,\forall f\in\Eklc X,\; e\lesssim f \Rightarrow e\leqq f.\]
For the third and last step, we show that for any expression $f\in\Eklc X$, there is an expression $[f]\in\Eklcm X$ such that $[f]_A\leqq f$ and whenever $e\lesssim f$ for $e\in\Eklcm X$ we have $e\lesssim [f]$.
This is enough to conclude thanks to \Cref{thm:completeness-klcm}.

In the next three subsections, we introduce constructions and prove lemmas necessary for each step.
Then, in \Cref{sec:main-theorem} we put them all together to show the main result.
\subsection{First step: normal forms}
\label{sec:nf}

A normal form is either an expression of the shape $\test A$ or of the shape $\test A\cdot e$ with $e\in\Eklcm X$.
We denote by $\NF$ the set of normal forms.
The main result of this section is the following:
\begin{lemma}\label{lem:nf}
  For any $e\in\Eklc X$ there exists a finite set $\nf e\subseteq\NF$ such that
  $e\equiv \sum_{\eta\in\nf e}\eta$.
\end{lemma}
\begin{proof}
  We show by induction on $e$ how to build $\nf e$.
  The correctness of the construction is fairly straightforward, and is left as an exercise : we will only state the relevant proof obligations when appropriate.

  For constants, variables, and unions, the choice is rather obvious:
  \begin{align*}
    \nf 0 \eqdef \emptyset && \nf 1 \eqdef \set{\test\emptyset} &&\nf x\eqdef\set{\test\emptyset\cdot x}&&
    \nf {e+f}\eqdef\nf e\cup \nf f.
  \end{align*}
  The case of mirror image is also rather straightforward:
  \[\nf{\conv{e}}\eqdef\setcompr{\test A}{\test A\in\nf e}\cup\setcompr{\test A\cdot \conv{e'}}{\test A\cdot e'\in\nf e}.\]

  For concatenations, we define the product $\eta\odot \gamma$ of two normal forms $\eta,\gamma\in\NF$ as:
  \begin{equation*}
    \test A\odot\test B\eqdef \test{A\cup B}\quad\quad
    \test A\odot\test B\cdot e\eqdef \test A\cdot e\odot\test B\eqdef \test{A\cup B}\cdot e\quad\quad
    \test A\cdot e\odot\test B\cdot f\eqdef \test{A\cup B}\cdot \paren{e\cdot f}.
  \end{equation*}
  We then define $\nf{e\cdot f}\eqdef\setcompr{\eta\odot \gamma}{\tuple{\eta,\gamma}\in\nf e\times \nf f}$.
  For correctness of the construction, we would have to prove that
  $\forall \eta,\gamma\in\NF,\;\eta\cdot \gamma\equiv \eta\odot\gamma$.

  For intersections, we define $\otimes:\NF\times\NF\to\fpset\NF$:
  \begin{align*}
    \test A\otimes\test B&\eqdef \set{\test{A\cup B}}
    &
    \test A\otimes\test B\cdot e&\eqdef
    \test A\cdot e\otimes\test B\eqdef \setcompr{\test{A\cup B\cup C}}{C\in \interone(e)}\\
    \test A\cdot e\otimes\test B\cdot f&\eqdef \test{A\cup B}\cdot \paren{e\cap f}.
  \end{align*}
  We then define
  $\nf{e\cap f}\eqdef\bigcup_{\tuple{\eta,\gamma}\in\nf e\times \nf f}{\eta\otimes \gamma}$.

  Finally, for iterations we use the following definition:
  \[\nf{e^+}\eqdef\setcompr{\test A}{\test A\in\nf e}
  \cup\setcompr{\test {{\cup}_i A_i}\cdot \paren{\sum_i{e_i}}^+}{\setcompr{\test {A_i}\cdot e_i}{i\leqslant n}\subseteq\nf e}.\qedhere\]
\end{proof}
\begin{remark*}
  In \cite{andrekaEquationalTheoryKleene2011}, a similar lemma was proved (Lemma~3.4).
  However, the proof in that paper is slightly wrong, as it fails to consider that cases  $\test A\cap\test B$ (easy) and $\test A\cap\test B\cdot e$ (more involved).
\end{remark*}
\subsection{Second step: removing tests on the left}
\label{sec:rm-one-left}
Here we want to transform an inequation $\test A\cdot e\lesssim f$, into one one the shape $e\lesssim \reduce A f$, while maintaining that $\test A\cdot\reduce A f\leqq f$.
The construction of $\reduce A f$ is fairly straightforward, the intuition being that $\test A$ forces us to only consider interpretations such that $a\in A\Rightarrow \epsilon\in\sem a$.
Therefore, for any $a\in A$ we replace in $f$ every occurrence of $a$ with $1+a$.
\begin{lemma}\label{lem:reduce-ax}
  $\test A\cdot\reduce A f\leqq f\leqq\reduce A f$.
\end{lemma}
\begin{proof}
  Since $a\leqq 1+a$, we can show by induction that $f\leqq\reduce A f$.
  Also, if $a\in A$:
  \begin{align*}
    \test A\cdot (1+a)
    &\equiv\test A + \test A\cdot a\tag*{By \eqref{ax:seq-1} and \eqref{ax:seq-plus}}\\
    &\equiv\test A \cdot \test A + \test A\cdot a\tag*{By \eqref{eq:5}}\\
    &\equiv\test A\cdot (\test A + a)\tag*{By \eqref{ax:seq-plus}}\\
    &\equiv\test A\cdot a\tag*{By \eqref{eq:6}}.
  \end{align*}
  This proves for the case of variables that $\test A\cdot\reduce A f\leqq f$, and can be generalised to arbitrary expressions by a simple induction.
\end{proof}

For the other property, we rely on the following lemma:
\begin{lemma}\label{lem:reduce-lang}
  Let $\Xi$ be some alphabet, and $\sigma: X\to\Lang\Xi$ be an interpretation such that $\forall x\in X,\;\epsilon\notin\sigma(x)$.
  Then $\sem{\reduce A f}=\sem[\tau]{\reduce A f}$, where
  $\begin{array}[t]{rcl}
     \tau:X&\to&\Lang\Xi\\
     x&\mapsto&\sigma(x)\cup\setcompr{\epsilon}{x\in A}.
   \end{array}$
\end{lemma}
\begin{proof}
  The result follows from a straightforward induction, the only interesting case being that of variables $x\in A$.
  This case is a simple consequence of our definitions: \[\sem[\tau]{1+a}=\epsilon\cup\tau(a)=\epsilon\cup\sigma(a)\cup\epsilon=\epsilon\cup\sigma(a)=\sem{1+a}.\qedhere\]
\end{proof}

\begin{corollary}\label{cor:reduce-spec}
  Let $\tuple{A,e}\in\fpset X\times \Eklcm X$ such that $\test A\cdot e\lesssim f$, then $e\lesssim \reduce A f$.
\end{corollary}
\begin{proof}
  Since by \Cref{lem:reduce-ax} we have $f\leqq \reduce A f$ by soundness and transitivity of $\lesssim$ we have $\test A\cdot e\lesssim \reduce A f$.
  We want to show that $e\lesssim \reduce A f$, so by \Cref{cor:rm-nil} we only need to check that for any interpretation $\sigma:X\to\Lang\Sigma$ such that $\epsilon\notin\bigcup_{x\in X}\sigma(x)$ we have $\sem e \subseteq \sem {\reduce A f}$.
  If we take $\tau$ like in \Cref{lem:reduce-lang}, we get that 1) since for every variable $\sigma(x)\subseteq \tau(x)$, $\sem e\subseteq \sem[\tau]e$ and 2)
  since for every $a\in A$ we have $\epsilon\in\tau(a)$, we get $\sem[\tau]{\test A}=\epsilon$.
  Together these tell us that $\sem e\subseteq\sem[\tau]e=\epsilon\cdot\sem[\tau]e=\sem[\tau]{\test A\cdot e}$.
  Since $\test A\cdot e\lesssim \reduce A f$ we know that $\sem[\tau]{\test A\cdot e}\subseteq\sem[\tau]{\reduce A f}$, and by \Cref{lem:reduce-lang} we know $\sem{\reduce A f}=\sem[\tau]{\reduce A f}$.
  We may therefore conclude that $\sem e\subseteq \sem[\tau]{\test A\cdot e}\subseteq\sem[\tau]{\reduce A f}=\sem{\reduce A f}$.
\end{proof}
\subsection{Third step: removing tests on the right}
\label{sec:rm-one-right}

This last step relies on \Cref{prop:rm-nil} and \Cref{lem:nf}.
\begin{lemma}\label{lem:positive}
  For any expression $f\in \Eklc X$, there exists a one-free expression $[f]\in\Eklcm X$ such that $[f]\leqq f$ and for any one-free expression $e\in\Eklcm X$ such that $e\lesssim f$ we have $e\lesssim [f]$. In other words, $[f]$ is the maximum of the set $\setcompr{e\in\Eklcm X}{e\leqq f}$.
\end{lemma}
\begin{proof}
  We define $[f]\eqdef\sum_{\test \emptyset\cdot f'\in\nf f}f'$.
  We can easily check that $[f]\leqq f$:
  \[[f]\equiv 1\cdot [f]=
    \test \emptyset\cdot\sum_{\test \emptyset\cdot f'\in\nf f}f'
    \equiv \sum_{\test \emptyset\cdot f'\in\nf f}{\test \emptyset\cdot f'}
    \leqq \sum_{\eta\in\nf f}{\eta}
    \equiv f.\]
  For the other property, we rely on \Cref{prop:rm-nil}.
  Assume $e\lesssim f$, we want to show that $e\lesssim[f]$.
  By \Cref{cor:rm-nil}, it is enough to check that $\sem e\subseteq\sem{[f]}$ for interpretations $\sigma$ such that $\forall x\in X,\epsilon\notin\sigma(x)$.
  Let $\sigma$ be such an interpretation, and $u$ some word such that $u\in\sem e$.
  Notice that the condition on $\sigma$ ensures that $\forall x\in X, 1\cap\sigma (x)=\emptyset$, hence $\sem{\test A}\neq\emptyset$ implies that $A=\emptyset$ by \Cref{lem:sem-test}.
  Also, because $\sigma(x)$ never contains the empty word and $e$ does not feature the constant $1$, $u$ must be different from $\epsilon$.
  Since $e\lesssim f$, we already know that $u\in\sem f$.
  By \Cref{lem:nf} and soundness, we know that there is a normal form $\eta\in\nf f$ such that $u\in\sem\eta$.
  Since $u\neq\epsilon$, $\eta$ cannot be a test: that would imply by \eqref{eq:test-sub-id} that $\eta\leqq 1$, hence $\sem\eta\subseteq\sem 1 = \epsilon$.
  Therefore we know that there is a term $\test A\cdot f'\in\nf f$ such that $u\in\sem{\test A\cdot f'}$.
  This means that $u\in\sem {f'}$ and $\epsilon\in\sem{\test A}$.
  As we have noticed before, this means that $A=\emptyset$.
  Thus we get $u\in\sem {f'}$ and $\test \emptyset\cdot f'\in\nf f$, which ensures that $u\in\sem{[f]}$.
\end{proof}

\subsection{Main theorem}
\label{sec:main-theorem}
We may now prove the main result of this paper:
\complrkl
\begin{proof}
  Since $e\equiv f\Leftrightarrow e\leqq f\wedge f\leqq e$ and $e\simeq f\Leftrightarrow e\lesssim f\wedge f\lesssim e$, we focus instead on proving that $e\leqq f\Leftrightarrow e\lesssim f$.
  By soundness we know that $e\leqq f\Rightarrow e\lesssim f$, so we only need to show the converse implication.

  Let $e,f\in\Eklc X$ such that $e\lesssim f$.
  By \Cref{lem:nf} we can show that $e\equiv \sum_{\eta\in\nf e}\eta$.
  Let $\eta\in\nf e$.
  Thanks to the properties of $\lesssim$  we have that $\eta\lesssim f$.
  There are two cases for $\eta$:
  \begin{itemize}
  \item either $\eta=\test A$ for some $A\in\fpset X$, in which case we have $\eta\leqq f$ by \Cref{lem:tests-compl};
  \item or $\eta=\test A\cdot e'$ with $A\in\fpset X$ and $e'\in\Eklcm X$.
    In that case, by \Cref{cor:reduce-spec} we have $e'\lesssim \reduce A f$, and by \Cref{lem:positive} we get $e'\lesssim[\reduce A f]$.
    Since both $e'$ and $[\reduce A f]$ are one-free, we may apply \Cref{thm:completeness-klcm} to get a proof that $e'\leqq [\reduce A f]$.
    Therefore
    \begin{align*}
      \eta=\test A\cdot e'\leqq \test A\cdot[\reduce A f]
      &\leqq \test A\cdot \reduce A f\tag*{By \Cref{lem:positive}.}\\
      &\leqq f\tag*{By \Cref{lem:reduce-ax}.}
    \end{align*}
  \end{itemize}
  In both cases we have established that $\eta\leqq f$, so by monotonicity we show that \[e\equiv \sum_{\eta\in\nf e}\eta\leqq  \sum_{\eta\in\nf e}f\leqq f.\qedhere\]
\end{proof}

\section{The ``top'' problem}
\label{sec:top}

In reversible Kleene lattices, union and intersection form a distributive lattice, and $0$ acts as both the unit of union and the annihilator of intersection.
All that is missing to get a bounded distributive lattice is the unit of intersection and annihilator of union, namely the constant $\top$, to be interpreted as the full language.
However, this turns out to be more complicated than one might think.

The first idea that comes to mind is to add the sole axiom $\top + e = \top$.
This axiom just says that for any expression $e\leqq \top$, and is enough to show that $e\cap\top\equiv \top\cap e\equiv e$.
It is obviously sound, so we get soundness of the resulting axiomatic equivalence.
This axiomatic equivalence can be reduced without too much difficulty to that of reversible Kleene lattices, thanks to the following remark:
\begin{remark}
  If we write $\Etop X$ for expressions with $\top$, let $\phi:\Etop X \to \Eklc {X+1}$ be the function that replaces every occurrence of $\top$ with $\paren{\sum_{a\in X+1}(a+\conv a)}^\star$.
  Then the following identity holds:
  $\forall e,f\in\Etop X,\;e\equiv f\Leftrightarrow \phi(e)\equiv\phi(f).$
\end{remark}

This same construction, when applied to expressions without intersections, yields a completeness proof. 
In the presence of intersection however it is not complete.
We illustrate this with two examples.
\begin{example}[Levi's lemma]
  Levi's lemma for strings~\cite{levi1944semigroups} states that whenever we have two factorisations of the same word, i.e. $u_1\,u_2=v_1\,v_2$, then either $\exists w,\;u_1=v_1\,w\wedge v_2=w\,u_2$ or $\exists w,\;v_1=u_1\,w\wedge u_2=w\,v_2$.
  If we now move from words to languages, it means that every word that can be obtained simultaneously as $L_1\cdot L_2$ and $M_1\cdot M_2$ also belongs to either $L_1\cdot\top\cdot M_2$ or $M_1\cdot\top\cdot L_2$.
  In other words, the following inequation holds:
  \[\paren{e_1\cdot e_2}\cap\paren{f_1\cdot f_2}\lesssim
    \paren{e_1\cdot\top\cdot f_2}+\paren{f_1\cdot\top\cdot e_2}.\]
  However this equation is not derivable.
  This law also contrasts with the properties we can observe in every fragment of this algebra that we have studied: in every case, if a term without $\star$ or $+$ is smaller than a term $e+f$, then it must be smaller than either $e$ or $f$.
  One can plainly see that it is not the case here.
\end{example}
\begin{example}[Factorisation]
  Another troubling example is the following:
  \[\paren{a\cdot b}\cap\paren{a\cdot c}\lesssim
    a\cdot\paren{\paren{\top\cdot b}\cap\paren{\top\cdot c}}.\]
  As before, this inequation is valid, but it is not derivable, and it does not involve unions.
  This suggests that the (in-)equational theory of languages with just the signature $\tuple{\cdot,\cap,\top}$ is already non-trivial.
  We postulate that the key to adding $\top$ to Kleene lattices lies with a better understanding the theory of this smaller signature.
\end{example}
\clearpage
\appendix



\bibliography{bibli}
\end{document}